\documentclass[twoside,12pt]{article} 
\usepackage[a4paper,margin=1in]{geometry}
\usepackage{times}
\usepackage{amsthm}
\usepackage{thmtools}
\usepackage{thm-restate}
\usepackage{amsmath}
\usepackage[utf8]{inputenc} 
\usepackage[T1]{fontenc}    
\usepackage{tikz}
\usepackage{url}
\usepackage{float} 
\usepackage{booktabs}       
\usepackage{amsfonts}       
\usepackage{nicefrac}       
\usepackage{microtype}      
\usepackage{lipsum}
\usepackage{fancyhdr}       
\usepackage{graphicx}       
\usepackage{todonotes}
\usepackage{comment}
\usepackage{hyperref}

\graphicspath{{media/}}     

\usetikzlibrary{lindenmayersystems,calc,decorations,decorations.pathmorphing,decorations.markings,shapes,shapes.geometric,arrows,backgrounds, patterns,decorations.pathreplacing, positioning}

\tikzset{snake it/.style={decorate, decoration=snake}}
\tikzstyle{vertex} = [circle, draw=black, fill=black, inner sep=0pt,  minimum size=5pt]
\tikzstyle{edgelabel} = [circle, fill=white, inner sep=0pt,  minimum size=15pt]

\definecolor{blue}{RGB}{0,82,147} 
\definecolor{red}{RGB}{202,033,063}
\colorlet{green}{green!50!black}

\usepackage{microtype} 

\definecolor{blue}{RGB}{0,82,147} 
\definecolor{red}{RGB}{202,033,063}

\colorlet{green}{green!50!black}

\usepackage{hyperref}       
\usepackage{url}            
\usepackage{booktabs}       
\usepackage{amsfonts}       
\usepackage{nicefrac}       
\usepackage{microtype}      
\usepackage{lipsum}
\usepackage{fancyhdr}       

\title{Fault-Tolerant Approximate Distance Oracles with a Source Set

}

\author{
  Dipan Dey \\
  Tata Institute of Fundamental Research \\
  Mumbai \\
  India\\
  \texttt{dipan.dey@tifr.res.in} \\
   \and
  Telikepalli Kavitha \\
Tata Institute of Fundamental Research \\
  Mumbai \\
  India\\
  \texttt{kavitha@tifr.res.in} \\
}

\usepackage{graphicx} 
\usepackage{todonotes}
\usepackage{comment}
\usepackage{hyperref}
\newcommand{\LL}{\mathcal{L}}

\newcommand{\LCA}{\mathsf{LCA}}

\newcommand{\wt}{\mathsf{wt}}
\newcommand{\Dist}{\mathsf{Dist}}

\newcommand{\Ball}{\mathsf{Ball}}
\newtheorem{theorem}{Theorem}[section]
\newtheorem{lemma}[theorem]{Lemma}

\theoremstyle{definition}
\newtheorem{definition}[theorem]{Definition}

\newtheorem{proposition}[theorem]{Proposition}
\newtheorem{remark}[theorem]{Remark}
\newtheorem{claim}{Claim}

\newenvironment{claimproof}[1][Proof]{%
  \begin{proof}[#1]%
}{%
  \end{proof}%
}

\begin{document}

\begin{center}
    \LARGE
    \textsc{Fault-Tolerant Approximate Distance Oracles with a Source Set}
\end{center}
\noindent\rule{\textwidth}{0.4pt}
\vspace{0.2cm}  

\begin{center}
    \begin{minipage}{0.45\textwidth}
        \centering
        \large
        Dipan Dey\\
        \normalsize
        Tata Institute of Fundamental Research\\
        Mumbai, India\\
        dipan.dey@tifr.res.in
    \end{minipage}%
    \hfill
    \begin{minipage}{0.45\textwidth}
        \centering
        \large
        Telikepalli Kavitha\\
        \normalsize
        Tata Institute of Fundamental Research\\
        Mumbai, India\\
        kavitha@tifr.res.in
    \end{minipage}
\end{center}
\vspace{0.2cm}  
\noindent\rule{\textwidth}{0.4pt}

\vspace{0.4cm}
\begin{center}
    
\textsc{Abstract}
\end{center}

\vspace{0.3cm}

\noindent
Our input is an undirected weighted graph $G = (V,E)$ on $n$ vertices along with a source set $S\subseteq V$. 
The problem is to preprocess $G$ and build a compact data structure such that upon query $Qu(s,v,f)$ where 
$(s,v) \in S\times V$ and  $f$ is any faulty edge, 
we can quickly find a good estimate (i.e., within a small multiplicative stretch) of the $s$-$v$ distance in $G-f$. 

The work of Bil{\`{o}} et al. (Algorithmica 2022) on multiple-edge fault-tolerant approximate shortest path trees implies
a compact oracle for the above problem with a stretch of at most 3 and with query answering time $O(\log^2 n)$.
We show a very simple construction of an $S\times V$ approximate distance oracle  with $O(1)$ query answering time; its
size is $\widetilde{O}(|S|n + n^{3/2})$ and multiplicative stretch is at most 5. A single-edge fault-tolerant $ST$-distance oracle 
from the work of Bil{\`{o}} et al. (STACS 2018) plays a key role in our construction. We also give a construction of a fault-tolerant 
$S \times V$ approximate distance oracle of size $\widetilde{O}(|S|n + n^{4/3})$ with multiplicative stretch at most 13 and as before, 
with $O(1)$ query answering time.

\vspace{0.8cm}

\section{Introduction}

The problem of computing distances between all pairs of vertices in a given graph $G = (V,E)$ with positive
edge weights is a fundamental problem in graph algorithms. The problem here is to preprocess $G$ and build a 
compact data structure (called a {\em distance oracle})
that can quickly answer distance queries for any pair of vertices. As is often the case with real-world networks like
routing networks or road networks, links may fail or roads may be temporarily blocked. Thus we have to allow for the case of
faulty edges. 
Since several links are unlikely to fail simultaneously, we consider the case of a single edge failure. 

Instead of recomputing distances from scratch for all pairs of vertices after an edge has failed, the problem is 
to build a resilient data structure that can answer distance queries between vertices after a single edge failure. 
Furthermore, we assume there is a specific set $S \subseteq V$ of sources, e.g., 
$S$ is a set of starting locations on a road network or $S$ is a set of source nodes in a routing network. 
Thus we are interested in distances only for pairs $(s,v) \in S \times V$. 
Suppose $|S| \ll n$, say $O(n^{\epsilon})$ for some $\epsilon \in (0,1)$.
Then it feels wasteful to build a fault-tolerant distance oracle that maintains distances for all pairs of vertices.
Another application is Vickrey pricing~\cite{HershbergerS01}, with the objective of determining, for every $(s,v) \in S \times V$ where $S$ is a given subset of $V$ and every edge~$f$, how much the distance from $s$ to $v$ increases if $f$ were to fail.

Thus the distance oracle has to process queries of the form $(s,v,f)$ where $s$ is the source, $v$ is 
the destination, and $f$ is the failed edge. 
Upon query $Qu(s,v,f)$, the oracle has to return the $s$ to $v$ distance in $G-f$, where
$G-f$ is the graph obtained by removing edge $f$ from $G$.
So our problem is the following. 
\begin{itemize}
    \item Preprocess $G$ and build a compact data structure that can quickly answer distance queries for any pair in 
    $S \times V$ when an edge fails. Hence our distance oracle has to answer queries $Qu(s,v,f)$ where $s\in S,  v\in  V$,  
    and $f$  is the failed edge.
\end{itemize}

This data structure is called a {\em single edge fault-tolerant sourcewise distance oracle}.  
As discussed below, for undirected unweighted graphs, a single edge fault-tolerant sourcewise 
exact distance oracle of size $\tilde{O}(n^{3/2}\sqrt{|S|})$ with $\tilde{O}(1)$ query time is known~\cite{GuptaS18}.
Our goal is to design more compact distance oracles (for sublinear sets $S$) in weighted graphs.
Furthermore, for the sake of space efficiency, we are ready to relax {\em exactness}. 
Thus the problem we consider is to design a compact fault-tolerant sourcewise {\em approximate} distance oracle. 

Recall that `sourcewise' captures the fact that we are interested in distances between pairs $(s,v) \in S \times V$. 
For any pair $(s,v) \in S \times V$ and $f \in E$, let $||sv \diamond f||$ denote the distance from $s$ to $v$ in $G-f$.
A fault-tolerant approximate distance oracle is said to have {\em multiplicative stretch} $\alpha$ if the distance 
$d_{G-f}(s,v)$ returned by the oracle on query $Qu(s,v,f)$ is 
sandwiched between the actual distance and $\alpha$ times the actual distance,
i.e., $||sv \diamond f|| \leq d_{G-f}(s,v) \leq \alpha \cdot ||sv \diamond f||$.
We show the following result.

\begin{restatable}{mytheorem}{FirstTheorem}
    \label{lem : One layer}
    Let $G = (V,E)$ be an undirected graph on $n$ vertices with positive edge weights. 
       For any $S \subseteq V$, a fault-tolerant sourcewise approximate distance oracle with multiplicative
       stretch at most~5 and size
   $\widetilde{O}(|S|n + n^{3/2})$ can be constructed in polynomial time such that 
   $Qu(s,v,f)$ where $(s,v) \in S\times V$ and $f \in E$ can be answered in constant time.
\end{restatable}

Note that our oracle has size $\widetilde{O}(n^{3/2})$ when $|S| = O(\sqrt{n})$. For smaller sets $S$, we show 
a sparser fault-tolerant sourcewise approximate distance oracle at the expense of 
a larger stretch. Its query
answering time is also $O(1)$.

\begin{restatable}{mytheorem}{SecondTheorem}
    \label{lem : two layers}
     Let $G = (V,E)$ be an undirected graph on $n$ vertices with positive edge weights. 
       For any $S \subseteq V$, a fault-tolerant sourcewise approximate distance oracle with multiplicative stretch at most~13
       and size $\widetilde{O}(|S|n + n^{4/3})$ can be constructed in polynomial time such that 
   $Qu(s,v,f)$ where $(s,v) \in S\times V$ and $f \in E$ can be answered in constant time.
\end{restatable}

Thus the above oracle has size $\widetilde{O}(n^{4/3})$ when $|S| = O(n^{1/3})$. 
The work of Bil{\`{o}}, Gual{\`{a}}, Leucci,  and Proietti on multiple-edge fault-tolerant approximate shortest path trees \cite{BiloGLP22} 
in undirected weighted graphs with a {\em single source} (so $|S| = 1$) implies a multiple-edge fault-tolerant {\em sourcewise} (so $S$ is any subset of $V$) approximate distance oracle of size $\widetilde{O}(|S|n)$ with a stretch of $3$.\footnote{Unfortunately, we were unaware of this work till very recently. We thank Manoj Gupta for bringing this paper to our attention.} 
Thus their oracle is sparser than our oracles when $|S|$ is small and it also achieves a better stretch.
However, our query answering time is $O(1)$ while theirs is $O(\log^2(n))$, where 
$n$ is the number of vertices. Our algorithms are truly simple while their techniques are quite involved.

As mentioned above, the problem of constructing fault-tolerant sourcewise exact distance oracles in undirected unweighted graphs has been studied earlier. Also, in undirected weighted graphs, the problem of constructing fault-tolerant single source exact distance oracles has been studied. We discuss these results below.

\subparagraph{Background.} 
The first fault-tolerant exact distance oracle was designed by Demetrescu and Thorup in 2002 \cite{DemetrescuThorup02} and it was for directed weighted graphs. Their oracle handles single edge failures and has size $O(n^2\log n)$ with $O(1)$ query time. After this result, there has been 
a long line of research on the problem of efficiently constructing single edge/vertex fault-tolerant exact distance oracles. Ignoring preprocessing time, the most space-efficient oracle is by Duan and Zhang~\cite{DuanZ17a} with size $O(n^2)$ and query time $O(1)$. Thus it shaves off the $\log n$ factor from the size of the original oracle.

\medskip

\noindent{\em Fault-tolerant sourcewise distance oracles.}
For undirected unweighted graphs,
Gupta and Singh~\cite{GuptaS18} designed a single edge fault-tolerant sourcewise exact distance oracle  of
$\tilde{O}(n^{3/2}\sqrt{|S|})$ size with $\tilde{O}(1)$ query time and source set $S$. 
In undirected graphs with edge weights in the range $\{1,2, \dots,M\}$, Bil{\`{o}}, Cohen, Friedrich and Schirneck~\cite{BiloC0S21} designed a fault-tolerant single source exact distance oracle. This oracle handles single edge failures and has size $\tilde{O}(n^{3/2}\sqrt{M})$ with query time $\tilde{O}(1)$. 
For undirected unweighted graphs, Dey and Gupta~\cite{DeyG22} designed a different oracle with the same space and query time bounds as in~\cite{BiloC0S21}, but with a faster preprocessing time. 

\medskip

\noindent{\em $ST$-distance oracles in directed graphs.}
In directed weighted graphs, 
Bil{\`{o}}, Choudhary, Gual{\`{a}}, Leucci, Parter and Proietti~\cite{BiloCG0PP18}
designed a fault-tolerant $ST$-distance oracle, i.e., it maintains exact distances for all pairs in $S \times T$,
for given vertex subsets $S$ and $T$. 
It handles single edge failures and has size $\widetilde{O}((|S|+|T|)n)$ with $O(1)$ query time,
where $n$ is the number of vertices.
They also designed a fault-tolerant $ST$-distance oracle in unweighted directed graphs
of size $\widetilde{O}(n \sqrt{|S||T|})$ with query time $O(\sqrt{|S||T|})$.
Furthermore, they showed a fault-tolerant  $ST$-{\em approximate} distance oracle in directed unweighted graphs 
that returns in constant time a distance estimate stretched by an additive term.
In particular, when $|S| = O(\sqrt{n})$, their oracle has size $\widetilde{O}(n^{3/2})$ and additive stretch $\widetilde{O}(\sqrt{n})$.

\subparagraph{Fault-tolerant approximate distance oracles.}
Approximate distance oracles that provide distances within a small multiplicative stretch for all vertex pairs have been extensively studied. Table~\ref{tab:example} summarizes results for fault-tolerant approximate distance oracles in directed/undirected graphs. 
Note that the stretch here is multiplicative, except for  the last row where the stretch has an additive term as well.

\begin{table}[h]
\centering
\small
\begin{tabular}{|p{1.75cm}|c|p{2.2cm}|p{3.7cm}|c|p{0.5cm}|} 
\hline
Graph & Faults & Stretch & Size & Query time & Ref \\
\hline
Undirected Weighted & $c\geq 1$ & $(8k-2)(c+1)$, \ \ \ \ \ $k\ge 1$ integer & $O(ckn^{1+1/k}\log(nM))$, $M$ is the max edge wt& $\tilde{O}(c)$ & \cite{ChechikLPR12} \\
\hline
Undirected \newline Unweighted & $c = 1$ & $(2k-1)(1+\epsilon)$, \ \ \ $k \ge 1$ integer and $\epsilon >0$ & $\tilde{O}\left(\frac{k^5}{\epsilon^4}n^{1+1/k}\right)$ & $O(1)$ & \cite{BaswanaK13} \\
\hline
Undirected Weighted & $c = o\left(\frac{\log n}{\log \log n}\right)$ & $(1+\epsilon)$& $O(n^2(\log D/\epsilon)^c c \log D)$& $O(c^5\log D)$ & \cite{ChechikCFK17} \\
\hline
Directed \newline Unweighted & $c\geq2$ & $(3+\epsilon)$& $\tilde{O}(n^{2-\frac{\alpha}{c+1}}/\epsilon)(\log n / \epsilon)^c)$ where $\alpha\in(0,1/2)$ and $\epsilon >0$ & $O(n^{\alpha}/\epsilon^2)$ & \cite{BiloCCC0KS24} \\
\hline

Undirected \newline Unweighted & $c = o\left(\frac{\log n}{\log \log n}\right)$ & $(\frac{k+1}{k})(1+\epsilon)$ with additive stretch of 2, \ \ \ $k \ge 1$ integer and $\epsilon \geq 0$ & $O\left(\frac{n^{2-\frac{\gamma}{(k+1)(c+1)}+o(1)}}{\epsilon^{c+2}}\right)$ where $\gamma \in \left(0,\frac{k+1}{2}\right)$  & $O(n^{\gamma}/\epsilon^2)$& \cite{BiloCCC0S24} \\
\hline
Undirected Weighted & $c = o\left(\frac{\log n}{\log \log n}\right)$ & $(2k-1)$ where $k\ge 1$ integer & $O(n^{1+\frac{1}{k}+\alpha+o(1)})$ where $\alpha \in (0, 1)$  & $O(n^{1+\frac{1}{k}-\frac{\alpha}{k(c-1)}})$& \cite{BiloCCFKS23} \\

\hline
\end{tabular}
\caption{A table listing the works related to fault-tolerant approximate distance oracles where $D$ is the diameter of the graph.}
\label{tab:example}
\end{table}

For single edge faults, note that Chechik, Langberg, Peleg, and Roditty~\cite{ChechikLPR12} showed an approximate distance oracle with stretch~12 and size $\widetilde{O}(n^{3/2})$ and another with stretch~28 and size $\widetilde{O}(n^{4/3})$. In comparison to this, Theorem~\ref{lem : One layer} shows a sourcewise approximate distance oracle with stretch~5 and size 
$\widetilde{O}(|S|n + n^{3/2})$ and Theorem~\ref{lem : two layers} shows a sourcewise approximate distance oracle with stretch~13 and size $\widetilde{O}(|S|n + n^{4/3})$. Thus for small sets $S$, our oracles are as sparse and 
have smaller stretch.
Note that for single faults and every $k \ge 1$, approximate distance oracles by Baswana and Khanna~\cite{BaswanaK13} are almost 
as sparse as the oracles in \cite{ChechikLPR12} and have significantly smaller stretch. However these oracles 
work only for unweighted graphs.  

It is an open problem if our construction can be generalized to work for all integers~$k$, in other words, to show a
sourcewise approximate distance oracle of size $\widetilde{O}(|S|n + n^{1+1/k})$ and stretch~$8k-3$ with $O(1)$ query answering time 
for $k \ge 3$.
Our results show such a construction for $k = 1,2$. Note that
the remaining approximate distance oracles in Table~\ref{tab:example} 
have superconstant query time, so our oracles cannot directly be compared with them.

\subparagraph{Our techniques.} Our algorithms are simple to describe and use the $ST$-distance oracle by 
Bil{\`{o}}, Choudhary, Gual{\`{a}}, Leucci, Parter and Proietti~\cite{BiloCG0PP18}. Their oracle uses
{\em landmark} vertices, i.e., vertices picked uniformly at random from the vertex set $V$
(originally used by Bernstein and Karger~\cite{BernsteinK08}).\footnote{To the best of our knowledge, the name `landmark'
vertices was first used by Dey and Gupta~\cite{DeyGuptaESA24}.}
The oracle in Theorem~\ref{lem : One layer} uses this $ST$-distance oracle for the given source set~$S$ and $T = S \cup {\cal L}$, 
where ${\cal L}$
is our landmark vertex set. The oracle in Theorem~\ref{lem : two layers} is based on the same idea, however there are two levels of sampling here:
so we have two landmark vertex sets ${\cal L}_2 \subseteq {\cal L}_1$. 
Theorem~\ref{lem : One layer} and Theorem~\ref{lem : two layers} are proved in Section~\ref{sec:approximate} and Section~\ref{sec:approx-2},
respectively. We discuss preliminaries in Section~\ref{sec:prelims} and conclude in Section~\ref{sec:conclusions}.

\section{Preliminaries}
\label{sec:prelims}
This section describes the notation that will be used in the rest of the paper and also gives a sketch of the $ST$-distance oracle from
\cite{BiloCG0PP18}. Our input is an undirected graph $G = (V,E)$
with positive edge weights as given by $\wt: E \rightarrow \mathbb{R}_+$. 
For any path $\rho$ in $G$:
\begin{itemize}
    \item let $||\rho||$ be the {\em length} of $\rho$, i.e., $||\rho|| = \sum_{e\in\rho} \wt(e)$;
    \item let $|\rho|$ be the {\em hop length} of $\rho$, i.e., the number of edges in $\rho$.
\end{itemize}

For any $(u,v) \in V \times V$, a shortest path between $u$ and $v$ is a path of minimum length between $u$ and $v$. 
We assume the shortest path between any two vertices in the graph
is unique. This property can be achieved by random perturbation of the given edge weights (e.g., see \cite{ParterP13}). 
The  property of unique shortest paths was also used in \cite{BernsteinK09,DeyGupta24,DuanR22,GuptaS18,HershbergerS01}. 
We denote the shortest path from $u$ to $v$ by $uv$. Thus $||uv||$ is the distance between $u$ and $v$ in~$G$ and $|uv|$ 
is the hop length between $u$ and $v$ in $G$.

\begin{itemize}
    \item Let $G-f = (V, E\setminus\{f\})$ be the graph obtained after deleting edge $f$ from the graph $G$. As in $G$, 
    we assume there is a unique shortest path between any pair of vertices in $G-f$.
    \item For any $(u,v) \in V \times V$ and $f \in E$, let $uv \diamond f$ be the shortest path between $u$ and $v$ in $G - f$. 
So $||uv\diamond f||$ is the distance between $u$ and $v$ in $G-f$.
\end{itemize}

The concept of {\em landmark} vertices will be key to our distance oracles. 

\begin{definition}[Landmark Vertex Set, $\LL$]
\label{def:landmark}
    Sample each vertex in $G$ independently with probability $p$.
    The selected set (call it $\LL$) of vertices is the landmark vertex set. 
\end{definition}

The probability $p$ in Definition~\ref{def:landmark} will be set
to different values in Section~\ref{sec:approximate} and Section~\ref{sec:approx-2}.
The following proposition on the landmark vertex set $\LL$ will be very useful to us.
\begin{proposition}
\label{lem:landmark}
With high probability, for any pair of vertices $u$ and $v$, if  $|uv| \ge \lfloor \frac{3\ln n}{p}\rfloor$
then there is at least one landmark vertex on $uv$.
\end{proposition}
\begin{proof}
Since each vertex in $G$ is sampled independently with probability $p$, for any pair of vertices $u$ and $v$, the probability that there is {\em no} landmark vertex on $uv$ is $(1-p)^k$ where $k = |uv|+1$ is the number of vertices on $uv$. 
Because 
$|uv| \ge \lfloor\frac{3\ln n}{p}\rfloor$, the probability that there is no landmark vertex on $uv$ is at most:
\[(1-p)^{\frac{3\ln n}{p}} \ \ \le \ \ \left(\frac{1}{e}\right)^{3\ln n} \ \ \le \ \ \frac{1}{n^3}.\]

Thus for any pair of vertices $u$ and $v$ with $|uv| \ge \lfloor\frac{3\ln n}{p}\rfloor$, the probability that 
there is no landmark vertex 
on $uv$ is at most $1/n^3$. Hence the probability that there is some pair $(x,y) \in V \times V$ with 
$|xy| \ge \lfloor\frac{3\ln n}{p}\rfloor$ such that there is no landmark vertex on $xy$ is at most ${n\choose 2}/n^3 \le 1/n$. 
Thus with probability at least $1-1/n$, it is the case that 
for every pair $(u,v)\in V \times V$ with $|uv| \ge \lfloor\frac{3\ln n}{p}\rfloor$, 
there is at least one landmark vertex on $uv$.
\end{proof}

Since each vertex in $G$ is sampled with probability $p$, 
the expected size of $\LL$ is $np$. We will set $p = n^{-\delta}$ for some $\delta \in (0,1)$ in Section~\ref{sec:approximate} and Section~\ref{sec:approx-2}. Thus with high probability,  we  will have $|\LL| \le 2np$ (by Chernoff bound). 

\begin{itemize}
    \item If either $|\LL| > 2np$ or there exists a pair of vertices $u, v$ with $|uv| \ge \lfloor \frac{3\ln n}{p}\rfloor$
such that there is no vertex of $\LL$ on $uv$ then we will repeat the step of sampling vertices and construct another 
landmark vertex set such that both these properties hold for the set obtained.
\end{itemize}

Thus we will assume that $|\LL| = O(np)$ and every pair of vertices
$u, v$ with $|uv| \ge \lfloor \frac{3\ln n}{p}\rfloor$ has
at least one vertex of $\LL$ on $uv$.
The expected number of trials to obtain a desired landmark set $\LL$ is $O(1)$. 

\subparagraph{Fault-Tolerant $ST$-Distance Oracle.}
We now briefly discuss the algorithm of Bil{\`o}  et al.\cite{BiloCG0PP18} to construct a fault-tolerant exact distance oracle in $G$
for pairs $(s,t) \in S \times T$, where $S \subseteq V$ and $T \subseteq V$ are part of the input.

Fix a pair $(s,t) \in S \times T$ and let $f = (a,b)$ be any edge on $st$. Let
$\ell$ and $\ell'$ be the two landmark vertices  on  $st$  closest to  $a$ and $b$ on $as$ and $bt$, respectively.
There are 3 cases with respect to the replacement path $st\diamond f$: (i)~$st\diamond f$ goes through $\ell$, (ii)~$st\diamond f$ goes through $\ell'$,
     (iii)~$st\diamond f$ goes through neither $\ell$ nor $\ell'$.

Their algorithm builds tables to deal with each of these cases. Figure~\ref{fig4} captures the main idea.
The following theorem from \cite{BiloCG0PP18} will be used in our algorithms.

        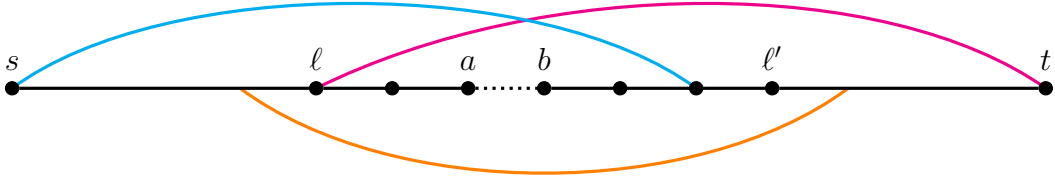
\begin{figure}[h]
\tikzstyle{vertex} = [circle, draw=black, fill=black, inner sep=0pt,  minimum size=5pt]
\tikzstyle{edgelabel} = [circle, fill=white, inner sep=0pt,  minimum size=15pt]
\centering
\pgfmathsetmacro{\d}{2}
\begin{minipage}{0.99\textwidth}
\begin{tikzpicture}[scale=1, transform shape]
               \path [very thick,draw,magenta] (-2,0) .. controls (1,1.5) and (5.5,1.5) ..  (7.6,0);
           \path [very thick,draw,cyan] (-6,0) .. controls (-4,1.5) and (1,1.5) ..  (3,0);
           \path [very thick,draw,orange] (-3,0) .. controls (-1,-1.5) and (3,-1.5) ..  (5,0);
\node[vertex, label=above:$a$] (a1) at (0,0) {};
\node[vertex, label=above:$b$] (b2) at (\d/2,0) {};
    \node[vertex] (b4) at (\d,0) {};
    \node[vertex] (b5) at (3*\d/2,0) {};
\node[vertex, label=above:$\ell$] (b3) at (-\d,0) {};
    \node[vertex] (a4) at (-\d/2,0) {};
    \node[vertex, label=above:$s$] (b1) at (-3*\d,0) {};
\node[vertex, label=above:$\ell'$] (a3) at ($(a1) + (2*\d, 0)$) {};
      \node[vertex, label=above:$t$] (a2) at ($(a1) + (3.8*\d, 0)$) {};
       \draw [very thick] (a2) -- (a3);
     \draw [very thick] (b1) -- (b3);
         \draw [very thick,dotted] (a1) -- (b2);
        \draw [very thick] (a1) -- (b3);
           \draw [very thick] (b2) -- (a3);

\end{tikzpicture}
\end{minipage}
\caption{The replacement path $st \diamond f$ where $f = (a,b)$ in case~(i) is $s\ell$ followed by the {\color{magenta} magenta} path $\ell t \diamond f$; in case~(ii) it is the {\color{cyan} blue} path $s\ell'\diamond f$ followed by $\ell't$ and in case~(iii) $st \diamond f$ avoids both $\ell$ and $\ell'$ - so the {\color{orange} orange} path is part of $st \diamond f$.}
\label{fig4}
\end{figure}

\begin{theorem}[\cite{BiloCG0PP18}]
\label{st-oracle}
    An $n$-vertex directed or undirected weighted graph G for given subsets $S$ and $T$ of $V$
can be preprocessed in polynomial time to compute a data structure of size $O((|S| + |T|)n\log n)$ 
that given any pair $(s,t) \in S\times T$ and any failing edge $f$ can report $||st \diamond f||$ in constant
time. 
\end{theorem}

\section{Sourcewise Approximate Distance Oracle for a Single Edge Fault}
\label{sec:approximate}

Our input is an undirected weighted graph $G = (V,E)$ with a positive weight function $\wt: E \rightarrow \mathbb{R}_+$ 
and a subset $S \subseteq V$ of sources. The goal is
to build a compact data structure that can answer distance queries $Qu(s,v,f)$ within a small multiplicative stretch,
where $s \in S, v \in V$ and $f$ is the edge fault.

\subparagraph{Landmark vertex set $\LL$.} Recall Definition~\ref{def:landmark} on landmark vertices. 
Let us sample each vertex independently with probability $p = (3\ln n)/\sqrt{n}$
to obtain our landmark vertex set~$\LL$.  
The following two properties hold (see Proposition~\ref{lem:landmark}); 
otherwise we resample to obtain another landmark vertex set 
$\LL$ so that the following two properties hold.
\begin{itemize}
    \item $|\LL| = O(\sqrt{n}\log n)$.
    \item For any pair of vertices $u$ and $v$: if  $|uv| \ge \lfloor\sqrt{n}\rfloor$,
then there is at least one landmark vertex on $uv$.
\end{itemize}

\subparagraph{Our algorithm.} On input $G = (V,E)$ and $S \subseteq V$, 
the first step of our algorithm is to build the above landmark vertex set $\LL$. We then compute
shortest path trees ${\cal T}(u)$ rooted at $u$ for all $u \in S \cup \LL$.
Along with every vertex $v$, the tree ${\cal T}(u)$ also has the two attributes $||uv||$ and $|uv|$, i.e., the length 
and the hop length of $uv$.

Our algorithm constructs the $ST$-exact distance oracle from \cite{BiloCG0PP18}
fixing the source set $S$ and destination set $T = S \cup \LL$. 
For each $v\in V$, let $t_v$ be the vertex in $T$ that is closest to~$v$, where
ties are broken arbitrarily. 
We maintain the lengths of replacement paths $vt_v \diamond f$ for each edge $f \in vt_v$.
Our algorithm is described below.

\begin{enumerate}
    \item Obtain the landmark vertex set $\LL$. 
\item For each $u \in S \cup \LL$ do: compute the shortest path tree ${\cal T}(u)$ rooted at $u$ in $G = (V,E)$.
    \item Use Theorem~\ref{st-oracle} to construct an $ST$-exact distance oracle for the given source set $S$ 
    and target set $T = \LL \cup S$ in $G = (V,E)$. 
     \item For every $v \in V$ in the graph $G = (V,E)$ do:
    \begin{itemize}
        \item Identify the nearest vertex to $v$ in the target set $T = \LL \cup S$. Call this vertex $t_v$.
        \item For $1 \le i \le |vt_v|$ do:
        \begin{itemize}
            \item Let $f_i$ be the $i$-th edge from $t_v$ on $vt_v$. 
            \item Compute the distance $||vt_v \diamond f_i||$ between $v$ and $t_v$ in $G-f_i$. 
            \item Set $\Dist_T[v,i] = ||vt_v \diamond f_i||$.
        \end{itemize}
    \end{itemize}
\end{enumerate}

\subparagraph{Query answering algorithm.} In response to the query $Qu(s,v,f)$, the query answering algorithm first checks if 
$f \in sv$. This check can be done efficiently via LCA queries. Given a rooted tree ${\cal T}$ and a pair of vertices $x,y$ in the tree 
${\cal T}$, recall that $\LCA_{\cal T}(x,y)$ is the least common ancestor of $x$ and $y$ in tree ${\cal T}$.

Observe that $f = (a,b) \in sv$ if and only if the answer to the following three questions is `yes' where ${\cal T}(s)$ is
the shortest path tree in $G$ rooted at $s$. 
\begin{itemize}
    \item Is $\LCA_{{\cal T}(s)}(v,a)$ equal to $a$?
    \item Is $\LCA_{{\cal T}(s)}(v,b)$ equal to $b$?
    \item Is $|sa| + 1 = |sb|$ or is $|sb| + 1 = |sa|$?
\end{itemize}

A `yes' answer to the first two questions implies that both $a$ and $b$ are vertices on the path $sv$.
Moreover, $a$ and $b$ are adjacent to each other on $sv$ if and only if the answer to the third question is `yes'.
Recall that for any vertex $w$, $|sw|$ is the hop length between $s$ and $w$, i.e., the number of edges in $sw$.

Given a tree ${\cal T}$, there is a linear time algorithm to build an $O(n)$ size data structure such that
LCA queries on ${\cal T}$  can be answered in $O(1)$ time~\cite{BenderF00}. Recall that for every vertex $w$,
the hop length $|sw|$ is stored along with $w$ in ${\cal T}(s)$. Thus $|sa|$ and $|sb|$ can be retrieved in $O(1)$ time.
Hence the query answering algorithm can determine in $O(1)$ time if $f \in sv$ or not. 
The query answering algorithm will return $||sv||$ if $f \notin sv$ (see Figure~\ref{fig2}).
Recall that the distance $||sv||$ is also stored along with $v$ in ${\cal T}(s)$.

\begin{figure}[h]
\tikzstyle{vertex} = [circle, draw=black, fill=black, inner sep=0pt,  minimum size=5pt]
\tikzstyle{edgelabel} = [circle, fill=white, inner sep=0pt,  minimum size=15pt]
\centering
\begin{tikzpicture}[scale=1, transform shape]      
       \coordinate (s) at (-5,0);
        \coordinate (t) at (5,0);
        \filldraw [black] (s) circle(3pt);
         \filldraw [black] (t) circle(3pt);
         \draw [very thick] (s) -- (t);

         \node[left, scale = 1.25, red] at (3,-0.5) {$a$};
         \node[left, scale = 1.25, red] at (3,-1.2) {$b$};
         \draw [ultra thick, red, dashed] (3,-0.5) -- (3,-1.2);
        \filldraw [red] (3,-0.5) circle(2pt); 
        \filldraw [red] (3,-1.2) circle(2pt); 
         \node[left, scale=1.25] at (s) {$s$};
         \node[right, scale=1.25] at (t) {$v$};
         
          \filldraw [black] (3,-1.8) circle(3pt);
           \node[right, scale=1.25] at (3,-1.9) {$t_v$};

           \path [very thick,draw,black] (3,0) -- (3,-0.5);
           \path [very thick,draw,black] (3,-1.2) -- (3,-1.8);
\end{tikzpicture}
\caption{Here $f = (a,b) \notin sv$, so the path $sv$ is undisturbed by the edge fault $f$.}
\label{fig2}
\end{figure}
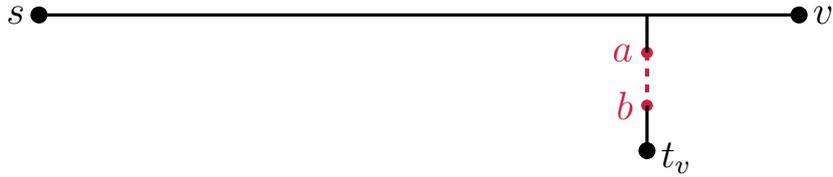

If $f \in sv$, then the query answering algorithm looks up the identity of $t_v$, which is the nearest
vertex in $T$ to $v$. 
Via LCA queries on ${\cal T}(t_v)$, we can determine if $f = (a,b) \in vt_v$ or not.
If not, then $||vt_v \diamond f|| = ||vt_v||$. So let us assume $f \in vt_v$.

Assume without loss of generality that $b$ is closer than $a$ to $t_v$,
i.e., $|at_v| = |bt_v| + 1$ (see Figure~\ref{fig2}). Let  $i$ be the index such that $a$ is the $i$-th vertex from $t_v$ on $vt_v$.
Then $||vt_v\diamond (a,b)|| = \Dist_T[v,i]$.
Recall that the attribute $i = |at_v|$ is stored along with $a$ in ${\cal T}(t_v)$.

\begin{itemize}
\item The query answering algorithm returns $||vt_v \diamond f|| + ||st_v \diamond f||$, where 
$||vt_v \diamond f|| = \Dist_T[v,i]$.
\end{itemize}

Note that the distance $||st_v \diamond f||$ is obtained by querying 
the $ST$-distance oracle. Thus the query answering time is $O(1)$.
We show below 
that $||vt_v \diamond f|| + ||st_v \diamond f||  \le 5||sv\diamond f||$.

\begin{lemma}
    \label{lem:source-correct}
    For any $(s,v) \in S \times V$ and $f \in E$, our algorithm returns an estimate for the $s$-$v$ distance in $G-f$ with a multiplicative stretch of at most $5$ in constant time.
\end{lemma}
\begin{proof}
It follows from the discussion above that the query answering time is $O(1)$.
Since the query answering algorithm will return $||sv||$ if $f \notin sv$ (see Figure~\ref{fig2}),
let us assume $f \in sv$. 
Then the distance estimate returned by the query answering algorithm in response to query $Qu(s,v,f)$  is 
$||st_v \diamond f|| + ||vt_v\diamond f||$ where $t_v$ is the nearest vertex in $T$ to~$v$. We now bound the sum
$||st_v \diamond f|| + ||vt_v\diamond f||$. 
Consider the following four cases.

        \begin{enumerate}     
        \item $f \in st_v$ and $f \in vt_v$. 
        This means the edge $f$ 
        belongs to the shortest path between $t_v$ and the least common ancestor of $s$ and $v$
        in ${\cal T}(t_v)$ (see Figure~\ref{fig2}). 
        However then $f \notin sv$, contradicting our assumption that $f \in sv$.

        \smallskip
        \item $f \in st_v$ and $f \notin vt_v$. The query answering algorithm 
        will determine via LCA queries that $f \notin vt_v$, so $||vt_v \diamond f|| = ||vt_v||$.
        Let us bound $||st_v \diamond f||$. The graph $G-f$ has an $s$-$t_v$ path 
        obtained by stitching the paths $sv\diamond f$ and $vt_v$, i.e., the path $sv\diamond f$ followed by $vt_v$.
        So $||st_v\diamond f|| \le ||sv\diamond f|| + ||vt_v||$. Hence the distance returned is at most $||sv \diamond f|| + 2||vt_v||$. 
        \begin{itemize}
        \item Because $t_v$ is the closest vertex in $T = \LL \cup S$ to $v$, 
        we have $||vt_v|| \le ||vs||$. Hence the distance returned is at most 
        $||sv\diamond f|| + 2||sv|| \le 3||sv\diamond f||$. So the stretch is at most~3 in this case. 
        \end{itemize}
        \smallskip
        \item $f \notin st_v$ and $f \in vt_v$. Since $f \notin st_v$, we have $||st_v \diamond f|| = ||st_v||$.
        Let us bound $||vt_v \diamond f||$. Since the graph $G-f$ has a $vt_v$ path obtained by stitching $sv\diamond f$ and $st_v$,
        we have $||vt_v \diamond f|| \le  ||sv \diamond f|| + ||st_v||$ (see Figure~\ref{fig3}). Thus the distance returned by the oracle is at most $||sv \diamond f|| + 2||st_v||$.
        \begin{itemize}
            \item Observe that $||st_v|| \le ||sv|| + ||vt_v|| \le 2||sv||$. Hence the distance returned is at most 
        $||sv\diamond f|| + 4||sv|| \le 5||sv\diamond f||$. Thus the stretch is at most 5 in this case.
        \end{itemize}

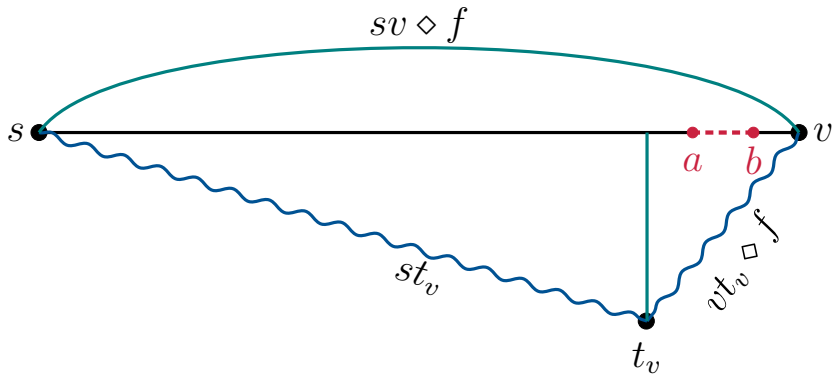
\begin{figure}[h]
\tikzstyle{vertex} = [circle, draw=black, fill=black, inner sep=0pt,  minimum size=5pt]
\tikzstyle{edgelabel} = [circle, fill=white, inner sep=0pt,  minimum size=15pt]
\centering
\begin{tikzpicture}[scale=1, transform shape]

       \coordinate (s) at (-5,0);
        \coordinate (t) at (5,0);
        \filldraw [black] (s) circle(3pt);
         \filldraw [black] (t) circle(3pt);
         \draw [very thick] (s) -- (3.6,0);
         \draw [very thick] (t) -- (4.4,0);

     \node[below, scale = 1.25, red] at (4.4,0) {$b$};
         \node[below, scale = 1.25, red] at (3.6,-0.1) {$a$};
          \filldraw [red] (4.4,0) circle(2pt); 
        \filldraw [red] (3.6,0) circle(2pt); 

         \node[left, scale=1.25] at (s) {$s$};
         \node[right, scale=1.25] at (t) {$v$};
         
          \filldraw [black] (2.99,-2.5) circle(3pt);
           \node[below, scale = 1.25] at (3,-2.6) {$t_v$};

      \draw[very thick, blue, decorate, decoration={snake, amplitude=0.5mm, segment length=5mm}] (s) -- (3,-2.5);
              \draw[very thick, blue, decorate, decoration={snake, amplitude=0.5mm, segment length=5mm}] (t) -- (3,-2.5);

           \path [very thick,draw,teal] (s) .. controls (-4,1.5) and (4,1.5) ..  (t);
           \path [very thick,draw,teal] (3,0) -- (3,-2.5);
           \draw (0,-1.9) node[ rotate=-15][scale=1.25] {$st_v$};
           \draw (3.9,-1.5) node[below, rotate=57][scale=1.25] {$vt_v \diamond f$};
            \draw (0,1.85) node[below][scale=1.25] {$sv \diamond f$};
             \draw [ultra thick, red, densely dashed] (4.4,0) -- (3.6,0);
\end{tikzpicture}
\caption{The oracle returns a distance estimate $\le 2||st_v|| + ||sv \diamond f||$, so the stretch is $\le  5$.}
\label{fig3}
\end{figure}
        
        \item $f \notin st_v$ and $f \notin vt_v$. The query answering algorithm will return $||st_v \diamond f|| + ||vt_v \diamond f|| = ||st_v|| + ||vt_v||$ in this case. We have $||vt_v|| \le ||vs||$ 
        and we also have $||st_v|| \le ||sv|| + ||vt_v|| \le 2||sv||$. Thus the stretch is at most 3 in this case.
    \end{enumerate}
This finishes the proof of the lemma.
\end{proof}

\subparagraph{Data structures constructed.} Our algorithm constructs in step~3 all the data structures constructed 
by the $ST$-distance oracle algorithm. Thus we have access to $||st \diamond f||$ for every $(s,t) \in S \times T$ and $f \in E$.
Our oracle also has the table $\Dist_T$ that stores $||vt_v \diamond f||$ between $v$ and $t_v$ in $G-f$, for each $v \in V$ 
and edge $f \in vt_v$. 
Let us bound the size of our oracle.
\newpage

\begin{lemma}
    \label{lem:source-size}
    The size of the data structures constructed by our algorithm is $\widetilde{O}(|S|n + n^{3/2})$.
\end{lemma}
\begin{proof}
    The size of $T$ is $\widetilde{O}(|S| + \sqrt{n})$. So the sizes of all the shortest path trees constructed in step~2 is
    $\widetilde{O}(|S|n + n^{3/2})$. Similarly,
    the size of the $ST$-oracle constructed in step~2 is $\widetilde{O}(|S|n + n^{3/2})$ (by Theorem~\ref{st-oracle}). 
    Furthermore, the data structure used to answer LCA queries on each shortest path tree ${\cal T}(u)$ has size $O(n)$. 
    Since $u \in S \cup \LL$,
    these 
    data structures also take up space $\widetilde{O}(|S|n + n^{3/2})$.
    
    It follows from the property of our landmark set $\LL$ that for any vertex $v$, we have 
    $\min\{|v\ell|: \ell \in \LL\} \le \lfloor\sqrt{n}\rfloor$.
    So for each vertex $v$, we have $||vt_v \diamond f||$ stored for at most $\lfloor\sqrt{n}\rfloor$ 
    many edges $f$, where $t_v$ is the nearest vertex in $T$ to $v$. 
    Thus the size of $\Dist_T$ is $O(n\cdot\sqrt{n}) = O(n^{3/2})$. 
    This finishes the proof of the lemma.
\end{proof}

It is easy to see that our algorithm runs in expected polynomial time.
Recall that we used randomization to construct the set $\LL$. By blowing up the size of $\LL$ by a factor of $\log n$, 
the construction of $\LL$ can be made deterministic (see \cite[Lemma~1]{BiloCG0PP18}).
Thus Theorem~\ref{lem : One layer} follows. We restate it below for convenience.
\FirstTheorem*

\begin{remark}
\label{remark1}
    The  query answering algorithm can return not only an approximate estimate of the 
    distance $||sv \diamond f||$, but also the corresponding approximate shortest 
    path in a succinct form.  It is known that any replacement path $\rho \diamond f$ is {\em 2-decomposable}, 
i.e., it is a concatenation of at most 2 shortest paths interleaved with at most 1 edge~\cite{afek-et-al-02}. 

So along
with any distance $||st_v \diamond f||$ (similarly, $||vt_v \diamond f||$), 
we could  also store the corresponding replacement paths in
2-decomposable form. Thus the query answering algorithm can return the corresponding $s$-$v$ approximate shortest path as the
union of two replacement paths $\rho_1 = st_v \diamond f$ and $\rho_2 = vt_v \diamond f$, each in 2-decomposable form, say,
$\rho_1 = \langle s,x,y,t_v\rangle$ and $\rho_2 = \langle v,x', y', t_v\rangle$.
This will mean $\rho_1$ is the shortest path in $G$ between $s$ and $x$ followed by the edge $(x,y)$, and the shortest path in $G$  between
$y$ and $t_v$, similarly for $\rho_2$.
\end{remark}

\section{A Sparser Fault-Tolerant Sourcewise Approximate Distance Oracle}
\label{sec:approx-2}

In this section, we present another fault-tolerant sourcewise approximate distance oracle for single edge faults.
As before, the input is an undirected weighted graph $G = (V,E)$ with a positive weight function $\wt: E \rightarrow \mathbb{R}_+$ 
and a subset $S \subseteq V$ of sources. Our goal is
to build a sparser data structure that can answer distance queries $Qu(s,v,f)$ within a small multiplicative stretch.
For sets $S$ of size $o(\sqrt{n})$,
the oracle in this section will be sparser than the one in Section~\ref{sec:approximate}.

\subparagraph{Our algorithm.} We will now construct {\em two} sets $\LL_1$ and $\LL_2$ of landmark vertices. 
We will first run the sampling step in Section~\ref{sec:prelims} with $p = (3\ln n)/n^{1/3}$. Let $\LL_1$
be the resulting landmark set. The following properties follow from Section~\ref{sec:prelims} (see Proposition~\ref{lem:landmark}).
\begin{itemize}
    \item $|\LL_1| = O(n^{2/3}\log n)$.
    \item For any pair of vertices $u$ and $v$: if $|uv| \ge \lfloor n^{1/3}\rfloor$ then there is at least one vertex of $\LL_1$ on $uv$.
\end{itemize}

After that, we sample each vertex of $\LL_1$ with probability $1/n^{1/3}$. Let $\LL_2$
be the set of selected vertices. 
Observe that this 2-step sampling to obtain $\LL_2$ is equivalent to running the sampling step in Section~\ref{sec:prelims} 
with $p = (3\ln n)/n^{2/3}$ on the entire vertex set $V$.
Thus the following properties follow from Section~\ref{sec:prelims} (see Proposition~\ref{lem:landmark}).
\begin{itemize}
    \item $|\LL_2| = O(n^{1/3}\log n)$.
    \item For any pair of vertices $u$ and $v$: if $|uv| \ge \lfloor n^{2/3}\rfloor$ then there is at least one vertex of $\LL_2$ on $uv$.
\end{itemize}

Rather than sampling each vertex of $V$ with probability $p = (3\ln n)/n^{2/3}$ to get $\LL_2$, we did this in two steps so that we have
$\LL_2 \subseteq \LL_1$.
Let $T_1 = \LL_1 \cup S$ and let $T_2 = \LL_2 \cup S$. Our algorithm will use the following notations for any vertex $v$.
\begin{itemize}
    \item Let $t_v$ be the vertex in $T_1$ that is nearest to $v$.
    \item Let $t'_v$ be the vertex in $T_2$ that is nearest to $v$.
\end{itemize}

For each $u \in T_2$, we will keep the shortest path tree ${\cal T}(u)$ in $G$ rooted at $u$.
However we cannot afford to keep shortest path trees rooted at each $u \in \LL_1$ since that would exceed the desired space bound.
Corresponding to each $u \in \LL_1$, let $\Ball(u) = \{v \in V: t_v =  u\}$ be the set of all vertices
    $v$ that regard $u$ as their nearest vertex in $T_1$.
    \begin{itemize}
        \item For each $v \in \Ball(u)$, we  will store the path $uv$.
        \item Thus we keep a {\em truncated} shortest path tree $\hat{\cal T}(u)  = \cup_{v\in \Ball(u)}uv$  
        in $G$ rooted at $u$ for each $u \in \LL_1$.
    \end{itemize}
    
    Along with each vertex $v \in \hat{\cal T}(u)$ where 
    $u \in \LL_1$, we also store $|uv|$, i.e., the hop length of~$uv$, and the distance $||uv||$. 
    Similarly, as done in Section~\ref{sec:approximate},
    along with each vertex $v \in {\cal T}(u)$, where $u \in T_2$, we store $|uv|$ and $||uv||$.

Below we describe the steps in our algorithm. 
\begin{enumerate}
    \item Obtain the landmark sets $\LL_1$ and $\LL_2$, where $\LL_2 \subseteq \LL_1$, as described above. 
\item For each $u \in T_2 = S \cup \LL_2$ do: compute the shortest path tree ${\cal T}(u)$ rooted at $u$ in~$G$.
    \item Use Theorem~\ref{st-oracle} to construct an $ST$-exact distance oracle for the given source set $S$ 
    and target set $T = T_2$ in $G = (V,E)$. 
 \item For each $u \in \LL_1$ do: compute the {\em truncated} shortest path tree $\hat{\cal T}(u)$  rooted at $u$ in~$G$.
     \item For every $v \in V$ do: 
    \begin{enumerate}
        \item Let $t_v \in T_1 = S \cup \LL_1$ be the vertex in $T_1$ that is nearest to $v$.
        \item For $1 \le i \le |vt_v|$ do: 
        \begin{itemize}
            \item Set $\Dist_1[v,i] = ||vt_v \diamond f_i||$ where $f_i$ is the $i$-th edge from $t_v$ on $vt_v$.
        \end{itemize}
    \end{enumerate}
    \item For every $u \in T_1$ do:
    \begin{enumerate}
        \item Let $t'_u \in T_2$ be the vertex in $T_2$ that is nearest to $u$.
        \item For $1 \le j \le |ut'_u|$ do:
        \begin{itemize}
        \item Set $\Dist_2[u,j] = ||ut'_u\diamond f_j||$ where $f_j$ is the $j$-th edge from $t'_u$ on $ut'_u$.
        \end{itemize}
    \end{enumerate}
\end{enumerate}

Observe that the array $\Dist_1[v,i]$ stores for any vertex $v$, the distance $||vt_v \diamond f||$ where $f$ is the
$i$-th edge from $t_v$ on the path $vt_v$.
Similarly, the array $\Dist_2[u,j]$ stores for any vertex $u \in T_1$, the distance $||ut'_u \diamond f||$ where $f$ is the
$j$-th edge from $t'_u$ on the path $ut'_u$.

\subparagraph{The query answering algorithm.}
In response to the query $Qu(s,v,f)$, the query answering algorithm first checks if 
$f \in sv$. As described in Section~\ref{sec:approximate}, this is done by checking the answers to some LCA queries
in ${\cal T}(s)$. Let us assume $f \in sv$, otherwise 
the query answering algorithm will return $||sv||$. 

Then the query answering algorithm looks up $x = t_v$ and $y = t'_x$.
In more detail, (i)~$x$ is the closest vertex to $v$ in $T_1$ and
(ii)~$y$ is the closest vertex to $x$ in $T_2$. 
The query answering algorithm needs to know if $f \in vx$ or not; if so, it also needs to know the index $i \in \{1,\ldots,n^{1/3}\}$ such that 
$f$ is the $i$-th edge on $xv$. As described in Section~\ref{sec:approximate}, we can decide if  $f \in vx$ or not via LCA
queries on the truncated shortest path tree $\hat{\cal T}(x)$. If so, we can also obtain from $\hat{\cal T}(x)$
the value $i$ such that 
$f$ is the $i$-th edge  from $x$ on $xv$.

Thus the query answering algorithm knows in $O(1)$ time 
whether $f \in xv$ or not and if so, the index $i$ such that $f$ is the $i$-th edge from $x$ on $xv$. 
\begin{itemize}
    \item If $f \notin vx$ then $||vx \diamond f|| = ||vx||$;
    else $||vx \diamond f|| = \Dist_1[v,i]$.
\end{itemize}

Recall that we compute ${\cal T}(u)$ for all $u \in T_2$. Thus, as described in Section~\ref{sec:approximate}, 
we can efficiently check if $f \in xy$ or not;  if so, the algorithm also knows
the index $j$ such that $f$ is the $j$-th edge from $y$ on the path $xy$. 
\begin{itemize}
   \item If $f \notin xy$ then $||xy \diamond f|| = ||xy||$;
   else $||xy\diamond f|| = \Dist_2[x,j]$.
\end{itemize}

Since $s \in S$ and $y \in T_2$ (recall that $T = T_2$), the distance $||ys \diamond f||$ is obtained by querying 
the $ST$-distance oracle. Thus the query answering algorithm can obtain $||vx\diamond f||,
||xy \diamond f||$, and $||ys \diamond f||$ in $O(1)$ time.
In response to the query $Qu(s,v,f)$, the query answering algorithm returns $||vx\diamond f|| + ||xy \diamond f|| + ||ys \diamond f||$. 

We will show in Lemma~\ref{lem:new-source-correct} that our
$s$-$v$ distance estimate in $G - f$ is at most $13||sv \diamond f||$. 

\begin{lemma}
    \label{lem:new-source-correct}
For any $(s,v) \in S \times V$ and $f \in E$, our algorithm returns an $s$-$v$ distance estimate with stretch $\le 13$
in $G-f$  in $O(1)$ time.    
\end{lemma}
\begin{proof}
Suppose the query is $Qu(s,v,f)$. If $f \notin sv$ then the algorithm returns $||sv||$, thus the stretch is 1 in this case.
So assume $f \in sv$. Then the query answering algorithm returns $||vx\diamond f|| + ||xy\diamond f|| + ||ys\diamond f||$, where 
$x = t_v$ and $y = t'_x$. 
Let us bound the stretch.

We need to compare the sum $||vx\diamond f|| + ||xy\diamond f|| + ||ys\diamond f||$
with $||sv\diamond f||$. Let us first show the following claim.
\begin{claim}
\label{claim1}
    We have (i)~$||vx|| \le ||vs||$, (ii)~$||xy|| \le 2||vs||$, and (iii)~$||ys|| \le 4||vs||$.
\end{claim}
\begin{claimproof}
    It follows from the definition of $T_1 = \LL_1 \cup S$ that both $x$ and $s$ are in $T_1$. Since $x$ is the nearest vertex in $T_1$ to $v$,
    we have $||vx|| \le ||vs||$.
    Recall that $y = t'_x$. Since $y$ is the closest vertex in $T_2$ to $x$, we have $||xy|| \le ||xt'_v||$, i.e., the $x$-$y$ distance is
    at most the distance between $x$ and $t'_v$ (recall that $t'_v$ is the nearest vertex in $T_2$ to $v$). 
    Furthermore, $||xt'_v|| \le ||xv|| + ||vt'_v||$. 
    
    Observe that both $||xv||$ and $||vt'_v||$ are
    at most $||sv||$ since $s \in T_1 \cap T_2$, so $v$'s distance to its nearest vertex in $T_1$ and also in $T_2$ is at most $||sv||$. 
    Thus $||xy|| \le 2||sv||$.
    So we have $||ys|| \le ||sv|| + ||vx|| + ||xy|| \le ||sv|| + ||sv|| + 2||sv|| = 4||sv||$.
\end{claimproof}

We are now ready to bound $||vx\diamond f|| + ||xy\diamond f|| + ||ys\diamond f||$.
There are 8 cases depending on the presence of edge $f$ on various shortest paths. 

\begin{enumerate}
    \item $f \notin vx$ and $f \notin xy$ and $f \notin ys$. 
    Then the algorithm returns $||vx|| + ||xy|| + ||ys||$.
    \begin{itemize}
        \item It immediately follows from Claim~\ref{claim1} that the $s$-$v$ distance estimate returned in this case is at most $7||sv|| \le 7||sv \diamond f||$. 
      \end{itemize}  
\medskip
    \item $f \notin vx$ and $f \notin xy$ and $f \in ys$. Then the algorithm returns  $||vx|| +
    ||xy|| + ||ys \diamond f||$. Since the failed edge $f$ belongs to neither $vx$ nor $xy$, 
    we have $||sy\diamond f|| \le ||sv\diamond f|| + ||vx|| + ||xy||$. 
    We have $||vx|| + ||xy|| \le 3||sv||$ (by Claim~\ref{claim1}).
    \begin{itemize}
        \item Thus the $s$-$v$ distance estimate returned in this case is at most 
        $||sv\diamond f|| + 6||sv|| \le 7||sv\diamond f||$ (by Claim~\ref{claim1}).
    \end{itemize} 
\medskip
    \item $f \notin vx$ and $f \in xy$ and $f \notin ys$. Then the algorithm returns $||vx||  + ||xy\diamond f||
    + ||ys||$. 
    Observe that $G-f$ has an $x$-$y$ path of length at most $||xv|| + ||vs \diamond f|| + ||sy||$. Since $||ys|| \le 4||vs||$
    (by Claim~\ref{claim1}),
    this $x$-$y$ path in $G-f$ is of length at most $||sv|| + ||sv\diamond f|| + 4||sv|| = 5||sv|| + ||sv\diamond f||$. 
    \begin{itemize}
        \item Using Claim~\ref{claim1} to bound $||vx||$ and  $||ys||$, the $s$-$v$
    distance estimate returned in this case is at most $||sv|| + 5||sv|| +  ||sv\diamond f|| + 4||sv|| = 10||sv|| + ||sv \diamond f|| \le 11||sv\diamond f||$.
    \end{itemize}
\medskip    
    \item $f \in vx$ and $f \notin xy$ and $f \notin ys$. Then the algorithm returns $||vx\diamond f||+||xy||
    + ||ys||$. Observe that $G-f$ has a $v$-$x$ path of length at most $||vs\diamond f|| + ||sy|| + ||yx||$. 
    This is of length at most $||sv\diamond f|| + 4||sv|| + 2||sv|| = 6||sv|| + ||sv\diamond f||$. 
    \begin{itemize}
        \item Using Claim~\ref{claim1} to bound $||xy||$ and $||ys||$,
    the $s$-$v$ distance estimate returned in this case is at most $||sv\diamond f|| + 6||sv|| + 2||sv|| + 4||sv|| =  12||sv|| + ||sv \diamond f|| \le 13||sv\diamond f||$.
    \end{itemize}

\medskip
 
   \item $f \notin vx$ and $f \in xy$ and $f \in ys$.  Consider the shortest path tree ${\cal T}(x)$ rooted at $x$ in $G$.
   Since $f \notin vx$ and $f \in xy$, the edge $f \in wy$ where $w = \LCA_{{\cal T}(x)}(v,y)$. 
   But the edge $f$ also belongs to $ys$ and $sv$ --- this is not possible (see Figure~\ref{fig5}). 
   Thus this case cannot arise.

\medskip

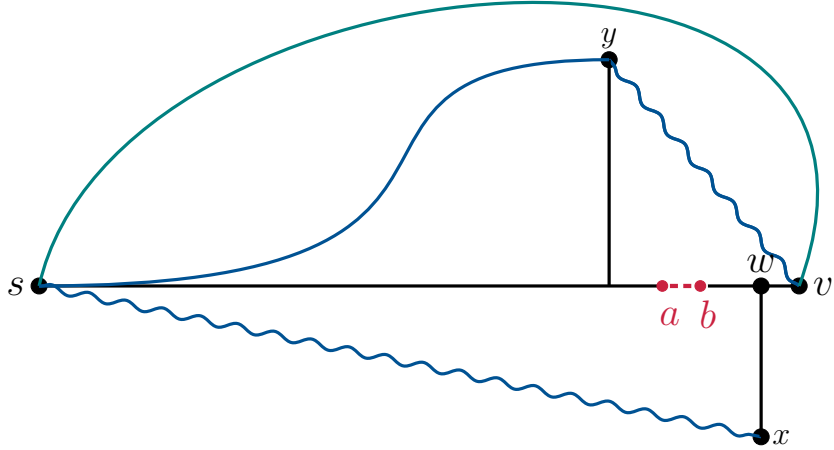
\begin{figure}[h]
\usetikzlibrary{decorations.pathmorphing}
\tikzstyle{vertex} = [circle, draw=black, fill=black, inner sep=0pt,  minimum size=5pt]
\tikzstyle{edgelabel} = [circle, fill=white, inner sep=0pt,  minimum size=15pt]
\centering
\begin{tikzpicture}[scale=1, transform shape]

       \coordinate (s) at (-5,0);
        \coordinate (t) at (5,0);
        \filldraw [black] (s) circle(3pt);
         \filldraw [black] (t) circle(3pt);
         \draw [very thick] (s) -- (3.2,0);
         \draw [very thick] (t) -- (3.8,0);
         \node[left, scale=1.25] at (s) {$s$};
         \node[right, scale=1.25] at (t) {$v$};

          \filldraw [black] (2.5,3) circle(3pt);
          \filldraw [black] (4.5,-2) circle(3pt);
           \node[above] at (2.5,3) {$y$};
           \node[right] at (4.5,-2) {$x$};
           
           \draw[very thick, blue, decorate, decoration={snake, amplitude=0.5mm, segment length=5mm}] (t) -- (2.5,3);
           
           \draw[very thick, blue, decorate, decoration={snake, amplitude=0.5mm, segment length=5mm}] (t) -- (2.5,3);
           
           \draw[very thick, blue, decorate, decoration={snake, amplitude=0.5mm, segment length=5mm}] (s) -- (4.5,-2);

           \path [very thick,draw,teal] (s) .. controls (-4,4.5) and (7,5.5) ..  (t);
           \path [very thick,draw] (2.5,0) --  (2.5,3);
           \path [very thick,draw] (4.5,0) --  (4.5,-2);
           
           \path [very thick,draw] (4.5,-2) --  (4.5,-2);
       
             \draw [very thick, blue] (2.5,3) .. controls (-2,3) and (2,0) .. (-5,0);

             \node[above, scale = 1.25] at (4.5,0) {$w$};
              \filldraw [black] (4.5,0) circle(3pt);        
                    \node[below, scale = 1.25, red] at (3.3,-0.1) {$a$};
         \node[below, scale = 1.25, red] at (3.8,0) {$b$};
           \draw [ultra thick, densely dashed, red] (3.3,0) -- (3.8,0);
                     \filldraw [red] (3.2,0) circle(2pt);
          \filldraw [red] (3.7,0) circle(2pt);

\end{tikzpicture}
\caption{The edge $f = (a,b)$ is supposed to be in the paths $sv, xy$, and $ys$, but not in $vx$. Here $w = \LCA_{{\cal T}(x)}(v,y)$ where 
${\cal T}(x)$ is the shortest path rooted at $x$ in $G$.}
\label{fig5}
\end{figure}

\medskip   
    \item $f \in vx$ and $f \notin xy$ and $f \in ys$. Consider  the shortest path tree 
    ${\cal T}(s)$ rooted at $s$ in $G$ and let $z = \LCA_{{\cal T}(s)}(v,x)$. 
    Since $f \in sv$ and $f \in vx$, it follows that $f \in zv$. Hence $f \notin sx$.
   Thus there is a $v$-$x$ path in $G-f$ of length $||vs\diamond f|| + ||sx||$. Since $||sx|| \le ||sv|| + ||vx|| \le 2||sv||$, this $v$-$x$ path in $G-f$ has length $||sv\diamond f|| + 2||sv||$. 
   
  \smallskip
   
   Since $f \in sv$ and $f \in sy$, the edge $f \in sr$ 
   where $r = \LCA_{{\cal T}(s)}(v,y)$.
   Thus the edge $f$ does not belong to the path $v$-$r$-$y$ in ${\cal T}(s)$. 
   Hence there is an $s$-$y$ path in $G-f$ of length 
   $||sv\diamond f|| + ||vs|| + ||sy|| \le ||sv\diamond f|| + ||vs|| + 4||sv||$ (by Claim~\ref{claim1}).
   \begin{itemize}
       \item Thus $G-f$ has  an $s$-$y$ path of length $||sv\diamond f|| + 5||vs||$, plus a $v$-$x$ path of length $||sv\diamond f|| + 2||vs||$.
   Since $||xy|| \le 2||vs||$, the $s$-$v$ distance estimate returned in this case is at most
   $2||sv\diamond f|| + 9||vs|| \le 11||sv \diamond f||$.
   \end{itemize}
   
\medskip       

    \item $f \in vx$ and $f \in xy$ and $f \notin ys$. As seen in case~6, 
    there is a $v$-$x$ path in $G-f$ of length $||sv\diamond f|| + 2||vs||$. 
    Moreover, since $f \in vx$ and $f \in xy$, the edge
   $f \in xw$ where $w = \LCA_{{\cal T}(x)}(v,y)$. Hence the edge $f$ does not belong to the path $v$-$w$-$y$ in ${\cal T}(x)$. 

\smallskip
    
    Thus there is a $v$-$y$ path in $G-f$ of length at most $||vx|| +||xy|| \le 3||sv||$. Hence there is an $x$-$v$-$y$ path of
    length at most $||sv\diamond f|| + 2||vs|| + 3||sv|| = ||sv\diamond f|| + 5||vs||$. 
    \begin{itemize}
        \item  So the $s$-$v$ distance estimate returned in this case is at most $||sy|| + (||sv\diamond f|| + 5||vs||) + (||sv\diamond f|| + 2||vs||)$. Since $||sy|| \le 4||sv||$, this is at most $2||sv\diamond f|| + 11||sv|| \le 13||sv\diamond f||$.
    \end{itemize}

\medskip

    \item $f \in vx$ and $f \in xy$ and $f \in ys$. As seen in case~7, there is a $v$-$x$ path in $G-f$ of length 
    $||sv\diamond f|| + 2||vs||$ and there is an $x$-$y$ path of length at most $||sv\diamond f|| + 5||vs||$. It also
    follows from case~7 that there is a $v$-$y$ path in $G-f$ of length at most $3||sv||$, thus there is an
    $s$-$y$ path in $G-f$ of length at most $||sv\diamond f|| + 3||vs||$. 
    \begin{itemize}
        \item  So the $s$-$v$ distance estimate returned in this case is at most 
    $3||sv\diamond f|| + 10||sv|| \le 13||sv\diamond f||$.
    \end{itemize}
\end{enumerate}
Thus the stretch of our approximate distance oracle is at most 13. We have already seen that the query answering time is $O(1)$. This finishes the proof of the lemma.
\end{proof}

 \subparagraph{Size of the oracle.}    
 We show below in Lemma~\ref{lem:new-source-size} that the space taken up by the
    data structures constructed in all the steps of our algorithm is
    $\widetilde{O}(n^{4/3} + |S|n)$.

\begin{lemma}
\label{lem:new-source-size}
The space needed to store all the data structures constructed by our algorithm is $\widetilde{O}(n^{4/3} + |S|n)$.
\end{lemma}
\begin{proof}
   The space taken by the truncated shortest path trees $\hat{\cal T}(u)$ for all $u \in \LL_1$ is $O(\sum_{v\in V}|vt_v|)$.
    Observe that $|vt_v| \le n^{1/3}$ (by Proposition~\ref{lem:landmark}). 
    Thus $O(\sum_v|vt_v|) = O(n^{4/3})$. 
    Similarly the space taken up by ${\cal T}(t)$ for all $t \in T_2$ is $O(n^{4/3}\log n + |S|n)$ since 
    $|T_2| = |\LL_2| + |S|$ and $|\LL_2|$ is $O(n^{1/3}\log n)$. 
   The size of the $ST$-oracle (where $T = T_2$) is also $\widetilde{O}(n^{4/3} + |S|n)$
   (by  Theorem~\ref{st-oracle}).
    
    For each vertex $v$, we store $t_v$ and $t'_v$ -- these are the nearest vertices to $v$ in $T_1$ and $T_2$, 
    respectively. 
  For all edges $f \in vt_v$, we store 
    $||vt_v\diamond f||$ in the data structure $\Dist_1$. We  have $|vt_v| \le n^{1/3}$ (by Proposition~\ref{lem:landmark}). 
    Thus the space taken by the data structure $\Dist_1$ to store the distances $\Dist_1[v,i]$ where $v \in V$ and $1 \le i \le n^{1/3}$  is at most $n^{4/3}$. 
    
    For all edges $f \in ut'_u$, where $u \in T_1$, the data structure $\Dist_2$ stores $||ut'_u\diamond f||$.
    For any vertex $u$, we have $|ut'_u| \le n^{2/3}$ (by Proposition~\ref{lem:landmark}). 
    Thus the space taken by $\Dist_2$ to store the distances $\Dist_2[u,i]$ where $u \in T_1$ and $1 \le i \le n^{2/3}$ is 
    $|T_1|\cdot n^{2/3} = O((n^{2/3}\log n + |S|)\cdot n^{2/3})$, which is $O(n^{4/3}\log n + |S|n^{2/3})$. Thus the entire
    space taken up by all the data structures is $\widetilde{O}(n^{4/3} + |S|n)$.
\end{proof}

It is easy to see that our algorithm runs in expected polynomial time.
As mentioned at the end of Section~\ref{sec:approximate}, by blowing up the sizes of $\LL_1$ and $\LL_2$ 
by a factor of $\log n$, their construction can be made deterministic as stated in \cite[Lemma~1]{BiloCG0PP18}.
Moreover, we can easily ensure that $L_2 \subseteq L_1$. Thus Theorem~\ref{lem : two layers} follows. 
We restate it below for convenience.
\SecondTheorem*

As mentioned in Remark~\ref{remark1}, along with every distance in $\Dist_1$ and $\Dist_2$, 
we could also store the corresponding
replacement paths in 2-decomposable form. Thus along with the approximate $s$-$v$ distance, 
the query answering algorithm can also return the
approximate path between $s$ and $v$ as the union of 3 paths, each in 2-decomposable form.

\section{Concluding Remarks}
\label{sec:conclusions}
Fault-tolerant approximate distance oracles that maintain approximate distances for all pairs of vertices 
have been well-studied. Fault-tolerant single source and multiple source 
{\em exact} distance oracles have also been studied. As mentioned in~\cite{BiloCG0PP18}, given a subset $S \subseteq V$, 
for the problem of storing $||sv \diamond f||$ where $(s,v) \in S \times V$ 
and $f \notin E$ is allowed\footnote{We thank a reviewer for pointing out this subtlety to us.} 
(this is interpreted the same as if no edge has failed), using standard tools, 
it can be shown that there are $n$-vertex graph families, for which any representation that
allows for the return of all the $S \times V$ post-failure distances must have size $\Omega(n^{3/2}\sqrt{|S|})$.
This motivates the study of sparser data structures
that maintain {\em approximate} distances for all pairs in $S \times V$ under the failure of any $f \in E$. 
Such a data structure is a fault-tolerant sourcewise approximate distance oracle.

We showed two such oracles with constant query answering time: one of size $\widetilde{O}(|S|n + n^{3/2})$ and stretch~5 and another of size $\widetilde{O}(|S|n + n^{4/3})$ and stretch~13.  
Upon query $Qu(s,v,f)$ where $f \notin E$, it turns out that both
our query answering algorithms return the original distance $||sv||$ as if no edge has failed.
There are several interesting open problems:
\begin{itemize}
    \item Are there approximate sourcewise distance oracles of size $\widetilde{O}(|S|n + n^{1+1/k})$ and stretch $8k-3$ with 
    $O(1)$ query answering time for all integers $k \ge 1$? Our constructions showed such oracles for $k = 1,2$.
    \item The study of fault-tolerant exact as well as approximate distance oracles has so far considered structured subsets of $V \times V$ such as
    $S \times T$. 
    Is there a sparse fault-tolerant exact or approximate distance oracle for an arbitrary subset ${\cal P}$ of $V \times V$?
\end{itemize}

\bibliographystyle{plain}  
\bibliography{references} 
\end{document}